\documentclass{article}
\usepackage{graphicx} 
\usepackage{amssymb}
\usepackage{stmaryrd}
\usepackage{mathrsfs}
\usepackage{graphicx}
\usepackage{amsmath}
\usepackage{amsthm}
\usepackage{caption}
\usepackage{subcaption}
\usepackage{listings}
\usepackage{verbatim}
\usepackage{braket}
\usepackage{subfiles}
\usepackage{bm}
\usepackage{tikz}
\usepackage{verbatim}
\usepackage[colorlinks]{hyperref}
\usepackage{color}
\usepackage{graphicx}

	\newtheorem{thr}{Theorem} 

	\newtheorem{cor}{Corollary}

	\newtheorem{prp}{Proposition}

  \newcommand{\lem}{\stackrel{\ast}{<}}
	\newcommand{\gem}{\stackrel{\ast}{>}}

    \newcommand{\ch}{\mathcal{H}}

	\newtheorem{clm}{Claim}

	\newcommand{\lea}{<^+}
	\newcommand{\gea}{>^+}
	\newcommand{\eqa}{=^+}
		
	\newcommand{\lel}{<^{\log}}
	\newcommand{\gel}{>^{\log}}

    \renewcommand{\H}{\mathbf{H}}

  \newcommand\N{\mathbb{N}}

 \newcommand{\K}{\mathbf{K}}
 \newcommand{\m}{\mathbf{m}}
 
 \newcommand{\I}{\mathbf{I}}
 \newcommand\BT{\{0,1\}}
 	\newcommand\FS{\BT^*}\newcommand\IS{\BT^\infty}
	
 \newcommand\ceil[1]{{\lceil#1\rceil}}
 
\title{\vspace*{-1cm}On Kolmogorov Structure Functions}
\author{Sam Epstein\footnote{samepst@icloud.com, www.jptheorygroup.com}}
\date{\today}
\begin{document}
\maketitle
\begin{abstract}
All strings with low mutual information with the halting sequence will have flat Kolmogorov Structure Functions, in the context of Algorithmic Statistics. Assuming the Independence Postulate, strings with non-negligible information with the halting sequence are purely mathematical constructions, and cannot be found in nature. Thus Algorithmic Statistics does not study strings in the physical world. This leads to the general thesis that two part codes require limitations as shown in the Minimum Description Length Principle. We also discuss issues with set-restricted Kolmogorov Structure Functions.
\end{abstract}
\section{Introduction}
In statistics, one tries to determine a model (such as a parameter for a distribution) from data which is assumed to have noise. In the Minimum Description Principle \cite{Grunwald07}, the model that describes information with the shortest code is assumed to be the best model. The data is described as a two part code, where the first part is the model and the second part is the noise. In one of his last works, Kolmogorov suggested a two part code for individual strings $x\in\FS$ based off Kolmogorov Complexity. The first part (the model) is a set $D$ containing $x$, the second part (the noise) is the code of $x$ given $D$, of size $\ceil{\log |D|}$. Other works examined probabilities and also total computable functions as models \cite{Vitanyi02}. Kolmogorov suggested the following \textit{structure function} at the Tallinn conference in Estonia, 1973.
$$
\H_k(x) = \min\{\log |S| : x\in S, \K(S)\leq k\}.
$$
The function $\K$ is the prefix Kolmogorov complexity. Theorem 1 of \cite{Vereshchagin2017} showed that any shape of the structure function is possible. This definition is used for the following function, which is a central definition of \textit{Algorithmic Statistics} \cite{Vereshchagin15,Vereshchagin2017,VereshchaginViSt04},
$$
k \mapsto k + \H_k(x)-\K(x).
$$
This function's equivalence to several other definitions is the main theorem of Algorithmic Statistics \cite{SemenovShVe24}. 

The structure function is flat for all strings with low mutual information with the halting sequence. Assuming the \textit{Independence Postulate}, \cite{Levin84,Levin13}, strings with non-negligible mutual information with the halting sequence are exotic, in that they cannot be found in nature. Such strings are purely mathematical constructions.
\section{Bounds}
\label{sec:bounds}

We review the results of \cite{GacsTrVi01}, in particular Theorem of III.24, which I don't think is widely known. $\m(x)$ is the algorithmic probability. The amount of information that the halting sequence $\ch\in\IS$ has about $x\in\FS$ is $\I(x;\ch)=\K(x)-\K(x|\ch)$. We use $x\lea y$, $x\gea y$ and $x\eqa y$ to denote $x< y+O(1)$, $x+O(1)>y$ and $x=y\pm O(1)$, respectively. In addition, $x\lel y$ and $x\gel y$ denote $x<y+O(\log y)$ and $x+O(\log x) > y$, respectively. Furthermore, ${\lem}f$, ${\gem}f$ denotes $< O(1)f$ and $>f/O(1)$.  For $x,y\in\FS$, $x\sqsubseteq y$ if $y=xz$ for some $z\in\FS$. $[A]=1$ if mathematical statement $A$ is true, and $[A]=0$ otherwise.

Let $S_k=\{x:\K(x)\leq k\}.$ Let $N_k = |S_k|$ where $\log N_k\eqa k - \K(k)$, due to \cite{GacsTrVi01}. Let $I^x_k$ be the index of $x$ in an enumeration of $S_k$. For $\K(x)=k$, let $m_x$ be the longest joint prefix of $I^x_k$ and $N_k$. So $m_x0\sqsubseteq I^x_k$ and $m_x1\sqsubseteq N_k$. Let $S_x = \{y:m_x0\sqsubseteq I^y_k\}.$ So
\begin{align*}
 \log |S_x| &\eqa k - \K(k) - \| m_x\| \\
 \K(S_x) & \lea \K(k) + \K(m_x) \lea \K(k) + \|m_x\| + \K(\|m_x\|).
\end{align*}
\begin{thr}[\cite{GacsTrVi01}]
\label{thr:main}
$$\|m_x\| <\K(\K(x))+\I(x;\ch)+O(\log \I(x;\ch)).$$
\end{thr}
\begin{proof}
    Let $\nu(y)= c [\K(y)\leq k]\m(y)2^{\|m_y\|}/(\|m_y\|^2)$. For proper choice of $c$, $\nu$ is a semimeasure and computable relative to $\ch$ and $k$. So $\K(x|\ch,k) \lea -\log\nu(x) \eqa \K(x) - \|m_x\| + 2\log \|m_x\|$.   
\end{proof}
Note that with some additional effort, the $\K(\K(x))$ term can be eliminated.
\begin{cor}
    For $x\in\FS$, $n=\K(x)$, for all $m\leq n$, $m\in\mathbb{W}$, there is a set $S\ni x$ such that $|S|=2^m$ and $\K(S) + m \lel n +\I(x;\ch).$ 
\end{cor}
\begin{clm}
    Thus there exists a set $S\ni x$ such that $\K(S) \lel 2\K(\K(x))+\I(x;\ch)$ and $\K(S) + \log |S|\lea \K(x) +\K(\K(x))+ O(\log(\I(x;\ch)+\K(\K(x))))$. This fact combined with the following proposition characterizes the Kolmogorov Structure Function.
\end{clm}
\begin{prp}
   Let $S\ni x$. For all $s<\log|S|$ there exists a set $S'\ni x$ such that $|S'|\leq |S|2^{-s}$ and $\K(S') \lea \K(S) + s + \K(s)$. 
\end{prp}\newpage
The minimal sufficient statistic for $x\in\FS$ is 
$$
\mathbf{k}^*(x)=\min \{ k:\H_k(x) + k = \K(x)\}.
$$
This is the location in which the Kolmogorov Structure Function reaches the boundary point and becomes flat. Due to Theorem \ref{thr:main}, $\mathbf{k}^*(x)\lel \K(\K(x))+\I(x;\ch)$ (but note that the $\K(\K(x))$ term can be eliminated). A visualization of the Kolmogorov Structure Function can be seen in Figure \ref{fig}.
\begin{figure}
	\begin{center}
		\includegraphics[width=0.8\columnwidth]{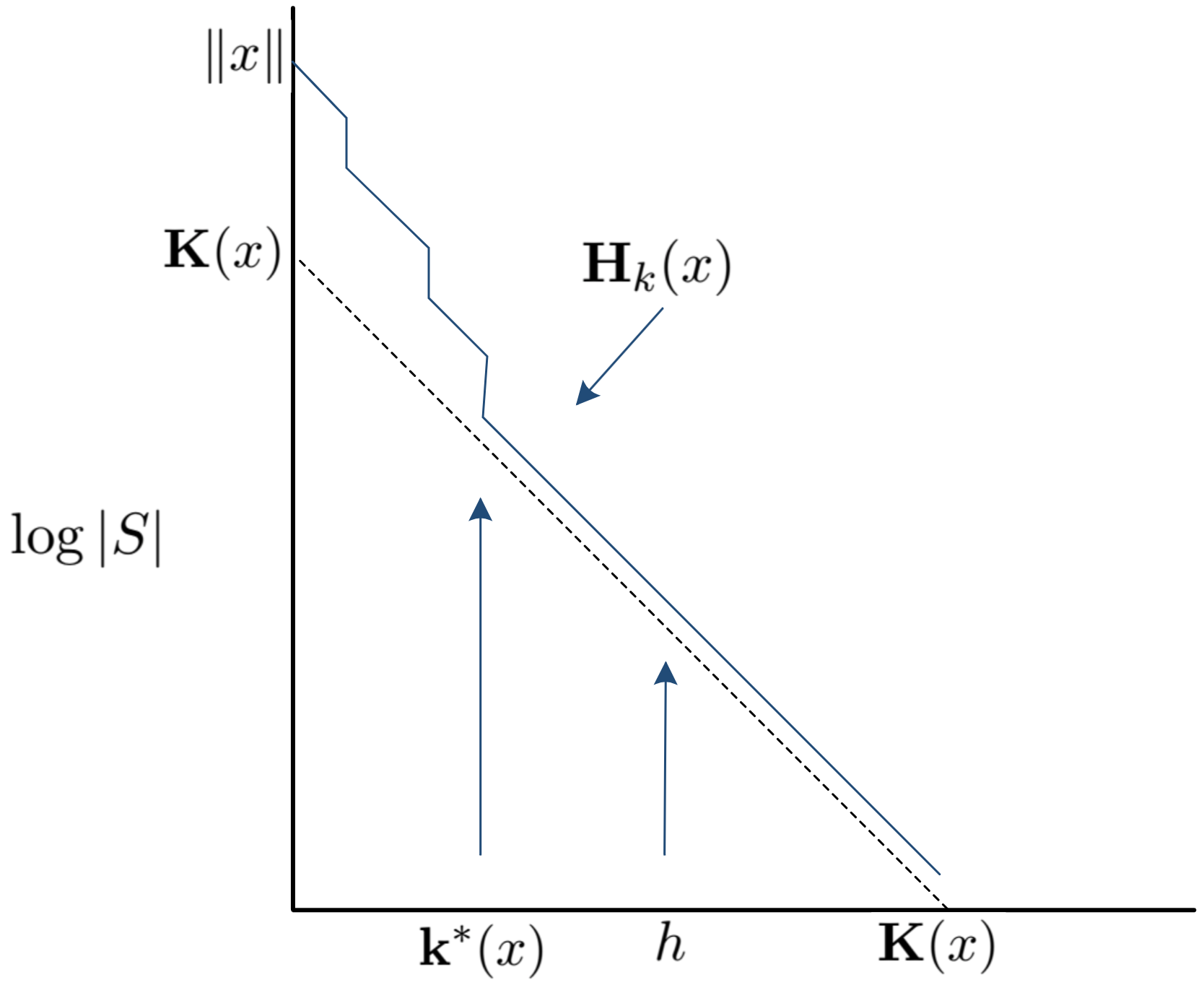}
		\caption{A visual representation of the Kolmogorov Structure Function $\H_k(x)$. The amount of information that the halting sequence has about $x$ is $h=\I(x;\ch)$. Since $h$ is negligible for almost all $x$, the Kolmogorov Structure Function is almost always flat.}
   \label{fig}
	\end{center}

 \end{figure}
\section{Set-Restricted Structure Functions}
\label{sec:setrest}
One potential method to create strings with non-simple Kolmogorov Structure Functions is to restrict the sets under consideration. Thus for a set of sets $\mathcal{S}$,
$$
\H_k^\mathcal{S}(x) = \min\{\log |S| : x\in S\in\mathcal{S},\K(S)\leq k\}.
$$
This would banish the pesky set $S_x$ defined in the last section. This was studied in Section 6 of \cite{Vereshchagin15}. However there is an inherent obstacle to proving such functions can have any shape. Proofs to statements (such as Theorem 10 in \cite{Vereshchagin15}) of such effect use 
a shape function $R$ to (non-recursively) construct a string $x$ whose structure function has that shape $R$ (up to a degree of precision depending on $\mathcal{S}$). Thus the proof can be thought of as a program to produce $x$ given $R$ and $\mathcal{H}$, with $\K(x|\ch)\lea \K(R)$. Thus proofs saying that for every shape $R$ there is a set $x$ such that $\H^\mathcal{S}_k(x)$ has shape $R$ (up to a certain precision) also implies that $\I(x;\ch)\gea\K(x)-\K(R)$. However, this obstacle does not preclude a proof of the existence of a large number of strings with profile $R$, which could potentially overcome the barrier described in this section.

In general, the Independence Postulate states if a string can be described by a small mathematical statement but has high Kolmogorov complexity then it cannot be found in the physical world. This presents an obstacle for constructive proofs in Algorithmic Information Theory.

\section{$\I(x;\ch)$ as an Error Term}
The Independence Postulate states one cannot find strings $x$ with nonnegligible $\I(x;\ch)$. Thus the term $\I(x;\ch)$ serves as a very good error term. Furthermore, $\I(x;\ch)$ enjoys the following deterministic and probabilistic conservation laws.\\

\noindent\textbf{Lemma.}
\begin{itemize}
\item \cite{Epstein22Exm22} For partial computable $f$, $\I(f(x);\ch)\lea \I(x;\ch)+\K(f)$.
\item \cite{EpsteinBirthday22} For probability $P$ over $\N$ computed by program $q$,\\ $\Pr_{a\sim P}[\I(a;\ch)>\I(q;\ch)+m] \lem 2^{-m}$.
\end{itemize}$ $\\
\vspace*{-0.6cm}

In addition, there are many provable statements about a mathematical construct $C$ with the following form
$$
\K(x(C)) \lel Q(C) + \I(C;\ch).
$$
The term $x(C)$ is some string associated with $C$. The term $Q(C)$ is some property about $C$. The term $\I(C;\ch)$ is the information $\ch$ has about the entire encoding of $C$. For example, as seen in, \cite{Epstein24}, let $C=\{(a_i,b_i)\}$ be a finite set of pairs of numbers, $x(C)$ be the simpliest total computable function consistent with $C$ and $Q(C)=\sum_i\K(b_i|a_i)$. One gets the following characterization of regression:\\

\noindent\textbf{Theorem.} $\K(x(\{(a_i,b_i)\})) \lel \sum_i\K(b_i|a_i)+\I(\{(a_i,b_i)\};\ch).$\newpage
\section{Minimum Description Length Principle}
\label{sec:mdl}
This Minimum Description Length Principle, \cite{Grunwald07}, is a principle to find regularity in information. When regularity is found, the data $D$ can be succinctly compressed. The set of permissible models is $\mathcal{M}$. The goal is to minimize the pair:
$$
\min_{M\in\mathcal{M}}L(M)+L(D|M).
$$
The term $L(M)$ is the length of the encoding of the model and the term $L(D|M)$ is the encoding of the data given the model. Typically, the set of permissible models $\mathcal{M}$ is severely limited and the encodings are efficiently computable. The term $L(D|M)$ can be thought of as the noise in the data $D$ given model $M$. Thus the expression represents tradeoff of the model complexity verses its descriptive power. The term $L(M)$ prevents overfitting of the data.

For example, take a very long string $x\in\FS$. A set of models $\mathcal{C}=\bigcup \mathcal{C}_k$ is all $k$th order Markov chains $\mathcal{C}_k$ on $\BT$. The term $L(M)$ for $M\in\mathcal{C}_k$ is all the parameters of a $k$th order Markov chain $M$. The term $L(x|M)$ is the negative logarithm of $x$ given $M$.

MDL is computable and has many practical applications whereas Algorithmic Statistics is a formal notion, providing theoretical results.
\section{Falsifiability}
In his book, \textit{The Logic of Scientific Discovery} \cite{Popper34}, the philosopher Karl Popper introduced the notion of falsiability, a deductive standard of evaluation of scientific theories and hypotheses. A theory or hypothesis (or in our case `model') is falsifiable if it can be contradicted by an empirical test. Popper proposed that falsiability is the indicator between scientific and non-scientific theories. 

For example, take meteorology and astrology. The complexity of astrology is not greater than the complexity of meteorology. Both theories fail in some of their predictions. However, consider the following assertion
\begin{quote}
    \textit{In the New York area, both a tropical storm and snowfall can happen in one hour.}
\end{quote}
According to the theory of meteorology, this is impossible. However, astrology does not preclude this possibility. Thus astrology is not falsifiable and not a scientific theory.

There is a connection between falsifiability and classification. A binary classification model $\mathcal{M}$ parameterized by $\theta$
shatters a set $S$ if for every binary assignment to the elements of $S$ there is a parameter $\theta$ that makes $\mathcal{M}$ completely consistent with the assignment. 

The VC dimension of $\mathcal{M}$ is the size of the largest set that is shattered by $\mathcal{M}$. The VC dimension can provide a probabilistic upper bound on the test error when using the model $\mathcal{M}$ on training data. Thus models that are not falsifiable will have an infinite VC dimension and no probabilistic upper bounds on the test error can be proved.

We now apply the notion of falsifiability to two part codes. Given a MDL pair $(\mathcal{H},L)$ consisting of a set of hypotheses $\mathcal{H}$ and a coding scheme $L$, the MDL estimator is of the form:
$$
\lambda_{x,\mathcal{H},L}(k)=\min\{L(H)+L(x|H):L(H)\leq k,H\in\mathcal{H}\}.
$$
From this definition, we get the following claim:
\begin{clm}
If MDL pair $(\mathcal{H},L)$ has corresponding function $\lambda_{x,\mathcal{H},L}$ that  reaches the $\K(x)$ line quickly for all $x$, then $(\mathcal{H},L)$ is not falsifiable.
\end{clm}

For example, in the example in Section \ref{sec:mdl}, the Markov model will be above the $\K(x)$ line for all 
 computationally simple prefixes $x$ of normal sequences. The story is different for the MDL pair $(\mathcal{S},K)$, where $\mathcal{S}$ is the set of all finite sets, $K(S)=\K(S)$, and $K(x|S)=[x\in S]\log|S|+[x\not\in S]\infty$. This pair is intimately connected to the structure function. Indeed, $\lambda_{x,S,K}$ is equal to the MDL estimator $\lambda_x$ in \cite{VereshchaginViSt04}.

For all $x\in\FS$ with $\I(x;\mathcal{H})\leq k$, $\lambda_{x,\mathcal{S},K}$ converges to the $\K(x)$ line at $\lel k$. Thus the pair $(\mathcal{S},K)$ is an optimal model for all (non-exotic) strings and is not falsifiable. Thus this pair is not a good theory for determining the structure of strings.

\section{Discussion}
The Independence Postulate \cite{Levin84,Levin13} states:

\begin{quote}
\textbf{IP}\textit{: Let $\alpha$ be a sequence defined with an $n$-bit mathematical statement (e.g., in PA or set theory), and a sequence $\beta$ can be located in the physical world with a $k$-bit instruction set (e.g., ip-address). Then $\I(\alpha:\beta)<k+n+c$ for some small absolute constant $c$.}
\end{quote}

When I first learned of \textbf{IP}, I didn’t realize how much of impact it could have on different fields of study. For example, \textbf{IP} and the Many Worlds Theory \cite{Everett57} are in conflict because measuring the spin of a million electrons results in the creation of a world where a large prefix of Chaitin’s Omega, $\Omega$, is found at a small address. Furthermore, \textbf{IP} causes issues in Constructor Theory \cite{Deutsch13}, which characterizes tasks in physics as either possible or impossible. This raises the question: “Is it possible or impossible to find large prefixes of $\Omega$?”. The answer causes trouble for either Constructor Theory or \textbf{IP}.

This note reiterates that \textbf{IP} implies Algorithmic Statistics does not study strings in the physical world. Thus the unrestricted structure function really doesn't say anything about good or bad models for a string. The set-restricted structure function might, but there are obstacles to showing this, as seen in Section \ref{sec:setrest}. This makes the connection between Algorithmic Statistics and the Minimum Description Length Principle \cite{Grunwald07} tenuous. 
This leads to a general thesis about separating strings into two part codes:
\begin{quote}
    \textit{Separating strings into two parts consisting of a model and noise requires substantial limitations on the group of models under consideration as well as their possible encodings.}
\end{quote}
The intention is not to denigrate the theory; a majority of my work (including \cite{EpsteinPhysics24,EpsteinOutliers23,Epstein23,Epstein24,EpsteinEL23}) is descendent from Algorithmic Statistics, particularly \cite{VereshchaginVi04}. My interpretation of the Kolmogorov Structure Function is that it (and its equivalent definitions) provide a means 
to characterize strings whose shortest programs have astronomically long running times. The Kolmogorov Structure Function (and its equivalent definitions) also provide a means to know that a string $x$ has high $\I(x;\ch)$.


\end{document}